\documentclass{article}
\usepackage{fullpage}
\usepackage{graphicx}
\usepackage{amsmath,amsfonts,amssymb}
\usepackage{amsthm}
\usepackage{physics}

\usepackage[pdfstartview=FitH,colorlinks,linkcolor=blue,citecolor=blue]{hyperref}

\theoremstyle{plain} 
\newtheorem{theorem}{Theorem}

\newtheorem{definition}[theorem]{Definition}
\newtheorem{proposition}[theorem]{Proposition}
\newtheorem{remark}[theorem]{Remark}
\newtheorem{lemma}[theorem]{Lemma}
\newtheorem{corollary}[theorem]{Corollary}
\newtheorem{conjecture}[theorem]{Conjecture}

\newcommand{\Deal}{{\mathsf{Deal}}}
\newcommand{\Recon}{{\mathsf{Recon}}}

\title{Quantum Secret Sharing and Error Correction vs No-Cloning}
\author{Steven Chien\\Chapter \and Ishani Mukherjee\\Stanford University \and Mark Zhandry\\Stanford University}
\date{}

\begin{document}

\maketitle

\begin{abstract}Secret sharing is ubiquitous throughout cryptography. All possible access structures are classically feasible, and in the case of threshold access structures, the protocols are even very efficient. However, when moving to the quantum setting, the no-cloning theorem shows that many access structures are impossible. In fact, no-cloning exactly characterizes which access structures are possible or impossible. For thresholds, quantum secret sharing is possible if and only if the threshold $t$ is strictly more than $n/2$.

In this work, we propose a variant of quantum secret sharing (QSS) where two or more identical copies of the input state are provided, but the output is only required to recover one copy. This notion circumvents the simple one-copy no-cloning obstruction, though the natural $k$-copy generalization still gives much milder obstructions, depending on how many copies are provided. For thresholds using $k$ copies, no-cloning implies that QSS is impossible whenever $t\leq n/(k+1)$.

Given that no-cloning exactly characterizes the one-copy case, it is tempting to hypothesize that no-cloning continues to exactly characterize the many-copy case. However, we show that this is not the case. We give positive results showing that multiple copies allow for going \emph{slightly} beyond the single-copy no-cloning obstruction. For thresholds, we construct QSS whenever $t>(n-k+1)/2$. On the other hand, we give a novel obstruction we call the Clique Path obstruction, which applies to arbitrary access structures. For thresholds, our obstruction shows that QSS is impossible whenever $t\leq (n-1)/k$. For all access structures, our upper and lower bounds exactly match for $k=2$. Both our results leverage connections to the $k$-colorability of certain graphs derived from the access structure. We leave closing the gap for $k\geq 3$ copies as a fascinating direction for future work.

Secret sharing is closely related to error correction, which can also be considered in the many-copy setting. Our results imply similar obstructions for quantum error correction for erasure channels.
\end{abstract}

\section{Introduction}

Quantum information fundamentally challenges many concepts from information theory and computation. On one hand, computers harnessing quantum physics promise significant computational speedups for certain problems~\cite{shor1997polynomial,Grover96}. On the other hand, the delicate and counter-intuitive nature of quantum states makes various tasks involving quantum information more challenging than their classical counterparts.

In this work, we focus on a concept called secret sharing~\cite{Shamir79}, which is closely related to certain kinds of error-correcting codes. Secret sharing allows for sharing a secret across several users, such that certain authorized sets of users can recover the secret, but unauthorized sets learn nothing about the secret. The collection of authorized sets is called the \emph{access structure}. Error correction allows for encoding information, such that small errors to the encoding can be corrected. Both secret sharing and error correction have been extensively studied in both the classical and quantum world, where in the quantum case one is interested in sharing/encoding quantum information. We focus exclusively on feasibility, and do not consider efficiency.

\paragraph{Quantum vs Classical Secret Sharing: The No-cloning Obstruction.} For \emph{classical} secret sharing (SS), \cite{ito1987secret} showed that secrets can be (inefficiently) shared to \emph{any} access structure $\Gamma$\footnote{A necessary and sufficient condition of being an access structure is that the collection of authorized user sets is monotone: adding a user to an authorized set results in a still-authorized set.}. An important special case is \emph{threshold} secret sharing (TSS), where any set of users of size at least a threshold $t$ can recover the secret. For threshold secret sharing, efficient schemes are possible for all numbers of users $n$ and thresholds $t$ such that $1\leq t\leq n$~\cite{Shamir79}.

Quantum secret sharing (QSS) and quantum threshold secret sharing (QTSS) are the natural analogs where the secret and protocol are now quantum. Unfortunately, the no-cloning theorem imposes a strong impossibility for various access structures/thresholds. In particular, if there are two disjoint authorized user sets in the access structure $\Gamma$, then there is no QSS scheme for $\Gamma$~\cite{cleve1999share,gottesman_solo}. Indeed, if such a QSS scheme existed, then we could clone any state by sharing it, and independently reconstructing it twice from the shares of the two authorized user sets. Specializing to thresholds, this means that QTSS is impossible if $t\leq n/2$. \cite{cleve1999share,gottesman_solo} complement these negative results by showing a feasibility result for any access structure $\Gamma$ other than those ruled out above; in particular, $t$-out-of-$n$ QTSS is feasible whenever $t>n/2$. Thus, the no-cloning theorem tightly characterizes what is feasible for quantum secret sharing.

Similarly, for erasure codes that correct lost registers, classical codes can handle an arbitrary constant fraction of erasures~\cite{reed1960polynomial} (using an appropriate rate), whereas the no-cloning theorem shows that the fraction of erasures for quantum codes must be strictly less than $1/2$~\cite{Grassl1997Codes}. No-cloning barriers also exist for correcting quantum Pauli errors; see Section~\ref{sec:discuss} for a more thorough discussion.

\paragraph{This work.} The limitations imposed by no-cloning lead us to ask the following natural question:
\begin{center}{\it 
Can the no-cloning obstruction be overcome by\\
using multiple copies of the state as input?}\end{center}

Intuitively, if provided two copies, then getting two copies of the state from the disjoint authorized sets would not violate the no-cloning theorem. 
Our work therefore initiates the study of quantum secret sharing and quantum error correction when multiple copies of the input are provided.

\subsection{Results Summary} 

We focus first on secret sharing, and then discuss the connection to error correction at the end of the overview.

We first fix some notation. In a $((n,t))$-QTSS scheme, the dealer takes as input $\ket{\psi}$, and produces $n$ shares of the secret, which are distributed to different users. Then any $t$ of the users can reconstruct the secret $\ket{\psi}$, while any fewer than $t$ users learn nothing about the secret. We generalize this concept to a $((n,t,k))$-QTSS scheme, where the dealer takes as input $k$ copies of $\ket{\psi}$. We only ask that a set of $t$ users recover at least one copy of $\ket{\psi}$; they do not need to recover all $k$ copies. 

We also consider QSS for general access structures $\Gamma$, in which case we use the notation $((n, \Gamma))$-QSS or $((n,\Gamma,k))$-QSS if the dealer is given $k$ copies of the secret $\ket{\psi}$.

\paragraph{The Generalized No-Cloning Obstruction.} We can readily generalize the no-cloning obstruction to the case of a general number of copies $k$: $((n,t,k))$-QTSS is impossible whenever $t\leq n/(k+1)$, since otherwise we can share the original $k$ copies and reconstruct $k+1$ copies. More generally, $((n,\Gamma,k))$-QSS is impossible whenever $\Gamma$ contains $k+1$ pairwise disjoint authorized subsets. This is because such a protocol would give a mechanism to convert $k$ copies into $k+1$ by running the dealer on $k$ copies and then independently reconstructing $k+1$ times using the disjoint authorized sets. Thus, we see that increasing $k$ is still subject to a generalized no-cloning obstruction, but the obstruction potentially allows for a richer set of access structures. Given that no-cloning exactly characterizes what is possible in the $k=1$ copy setting, a natural hypothesis is that no-cloning likewise characterizes the general $k$ setting.

Our next two sets of results show that (1) using more copies does allow for going beyond the limitations of a single copy, but (2) that there exist stronger obstructions than no-cloning, showing that the hypothesis above is false. Concrete examples to keep in mind are $((4,2,2))$-QTSS and $((5,2,2))$-QTSS, where using $k=2$ copies, $n=4$ or $n=5$ parties receive shares such that any $t=2$ of them can reconstruct. The no-cloning obstruction shows that $k>1$ copies are necessary, but potentially allows for $k=2$ copies.

\paragraph{Positive Results: Secret Sharing Leveraging Multiple Copies.} Our first main result is to utilize $k>1$ copies to bypass the 1-copy no-cloning obstruction (while of course still being subject to the generalized no-cloning obstruction). In the case of QTSS, the result is:

\begin{theorem}[Informal; see Theorem~\ref{thm:pos-bound}]\label{thm:pos-thresh-inf} There exists a $((n,t,k))$-QTSS scheme whenever $t> \frac{n-k+1}{2}$.
\end{theorem}

In particular, a $((4,2,2))$-QTSS scheme does indeed exist. Theorem~\ref{thm:pos-thresh-inf} is a special case of a more general theorem we prove for general access structures. For a general access structure $\Gamma$, our more general theorem makes use of the \emph{access structure graph} $G_\Gamma$ associated with $\Gamma$. The nodes of $G_\Gamma$ are the minimal sets in $\Gamma$, and two nodes are connected if the corresponding sets are \emph{disjoint}.

\begin{theorem}[Informal; see Theorem~\ref{thm:acc_comp_k_colorable}]\label{thm:pos-general-inf}There exists a $((n,\Gamma,k))$-QSS scheme if the associated access structure graph is $k$-colorable.\end{theorem}
In the case of thresholds, the graph $G_\Gamma$ is the Kneser graph $K(n,t)$, and Theorem~\ref{thm:pos-thresh-inf} follows from Theorem~\ref{thm:pos-general-inf} by observing that $K(n,t)$ is $k$-colorable whenever $t>\frac{n-k+1}{2}$. Also note that Kneser's Conjecture (formally proved in~\cite{lovasz1978kneser}) shows that $K(n,t)$ is \emph{not} $k$-colorable whenever $t\leq\frac{n-k+1}{2}$, tightly characterizing the colorability of $K(n,t)$.

Observe that Theorems~\ref{thm:pos-thresh-inf} and~\ref{thm:pos-general-inf} are loose with respect to the general no-cloning obstruction whenever $k>1$ (for $k=1$ they are tight by prior work). Indeed, a $((5,2,2))$-QTSS scheme is not excluded by the general no-cloning obstruction, but is also not guaranteed by Theorem~\ref{thm:pos-thresh-inf}. More generally, in terms of the access structure graph $G_\Gamma$, the general no-cloning obstruction shows that a $((n,\Gamma,k))$-QSS scheme cannot exist if $G_\Gamma$ contains a $(k+1)$-clique. While a $(k+1)$-clique implies $G_\Gamma$ is not $k$-colorable, the converse is not true.

\paragraph{Negative Results: Impossibility Beyond No-Cloning.} We also give non-trivial negative results beyond the no-cloning obstruction, which we call the \emph{Clique Path obstruction}. For thresholds, we show:

\begin{theorem}[Informal; see Theorem~\ref{thm:neg-allbig}]\label{thm:neg-thresh-inf} For $k\geq 2$, there does \emph{not} exist a $((n,t,k))$-QTSS scheme whenever $t\leq (n-1)/k$. 
\end{theorem}
In particular, in the case of $k=2$ copies, Theorem~\ref{thm:neg-thresh-inf} exactly matches the positive result in Theorem~\ref{thm:pos-thresh-inf}, showing that both are tight for $k=2$. For example, a $((5,2,2))$-QTSS scheme does \emph{not} exist.

Theorem~\ref{thm:neg-thresh-inf} follows from a more general impossibility for general access structures. Consider a graph $G$ and consider two $k$-cliques $C_0$ and $C_1$ which only differ in a single node. In this case, we say $C_0,C_1$ are \emph{adjacent}. A simple observation is that in any valid $k$-coloring, the colors of $C_0$ immediately give the colors of $C_1$. More generally, by chaining adjacent $k$-cliques $C_0,C_1,\cdots,C_\ell$ into a \emph{Clique Path}, we can conclude the colors of all $k$-cliques in the path from just the colors of $C_0$.

If $C_0$ and $C_\ell$ contain a common node, but the two cliques assign that node a different color, then we can conclude that $G$ is \emph{not} $k$-colorable. In this case, we say $G$ is \emph{non-$k$-colorable by Clique Paths}. Our general negative result shows that non-$k$-colorability by Clique Paths excludes any QSS scheme:

\begin{theorem}[Informal; see Theorem~\ref{thm:neg-criterion}]\label{thm:neg-general-inf} For $k\geq 2$, a $((n,\Gamma,k))$-QSS scheme does \emph{not} exist if $G_\Gamma$ is non-$k$-colorable by Clique Paths.
\end{theorem}

Very roughly, our proof is a proof by contradiction. We imagine artificial dealers associated to a coloring of a particular $k$-clique which, instead of using $k$ copies of a single state, use $k$ independent states $\ket{\psi_1}\cdots\ket{\psi_k}$, and deal shares in such a way that each authorized set in the $k$-clique recovers the state according to its color. We then show that dealers associated to adjacent $k$-cliques actually produce states that are identical, which extends to general Clique Paths by transitivity. If the access structure graph is non-$k$-colorable by Clique Paths, then the inconsistently-colored user recovers different states depending on the dealer, which contradicts the fact that the dealers all produce the same state.

Note that in the case $k=2$, a $2$-Clique Path assigning inconsistent colors is simply an odd-length cycle; the existence of such a cycle being equivalent to being non-2-colorable. Hence, Theorem~\ref{thm:neg-general-inf} and Theorem~\ref{thm:pos-general-inf} exactly match and are therefore tight in the case $k=2$.

Note that if $G$ has a $(k+1)$-clique, it is non-$k$-colorable by Clique Paths, but the converse is false. Hence, Theorem~\ref{thm:neg-general-inf} strictly generalizes the general no-cloning obstruction. On the other hand, there are graphs which are non-$k$-colorable, but are not non-$k$-colorable by Clique Paths; for example a graph may be triangle-free but still require more than 3 colors.

\paragraph{Application to Quantum Erasure Codes.} A $((d,n))$-Quantum Erasure
Code (QErC) encodes a quantum state $\ket{\psi}$ into $n$ registers such that
erasing any $d$ registers still allows for recovering $\ket{\psi}$. We can
naturally generalize to $((d,n,k))$-QErC where $k$ copies are given as input
but only a single copy needs to be recovered. No-cloning implies that
$d\geq n\frac{k}{k+1}$ is impossible. A $((d,n,k))$-QErC scheme is exactly a
$((n,n-d,k))$-QTSS scheme satisfying correctness but not necessarily security. Our Clique Path obstruction never invokes security (the alternate dealers are built from the reconstruction procedures alone), so setting $t=n-d$ in Theorem~\ref{thm:neg-thresh-inf} gives a stronger impossibility:

\begin{theorem}\label{thm:qerc}Regardless of rate, for $k\geq 2$, a $((d,n,k))$-QErC scheme does not exist whenever $d\geq n\frac{k-1}{k}+\frac{1}{k}$.
\end{theorem}





Our positive and negative results leave a gap for $k\geq3$. Given the central
role of $k$-colorings on both sides, it is natural to conjecture that our
positive results are exactly optimal; we state this as
Conjecture~\ref{conj} and discuss it in Section~\ref{sec:discuss}.

\subsection{Technical Overview}

\paragraph{Positive Results: The Union-of-Access-Structures Approach.} We prove Theorems~\ref{thm:pos-thresh-inf} and~\ref{thm:pos-general-inf} by making black-box use of $k$ separate single-copy solutions, which in turn are provided by~\cite{cleve1999share,gottesman_solo}. Consider the case $k=2$. Consider the assumed 2-coloring of the access structure graph $G_\Gamma$. Then let $\Gamma_0,\Gamma_1$ be the partition of $\Gamma$ according to the coloring. Recall that $G_\Gamma$ contains an edge when the sets are disjoint, in which case they are given different colors. Therefore, $\Gamma_0$ and separately $\Gamma_1$ have no pairs of disjoint sets, meaning a $((n,\Gamma_0,1))$-QSS scheme and a $((n,\Gamma_1,1))$-QSS scheme exist.

Therefore, to obtain a $((n,\Gamma,2))$-QSS scheme, we simply take the two copies of $\ket{\psi}$, and share one according to the $((n,\Gamma_0,1))$-QSS scheme, and the other copy according to the $((n,\Gamma_1,1))$-QSS scheme, giving each user their shares from both schemes. An authorized set of users $S$ checks their color -- that is, the $\Gamma_b$ that $S$ belongs to -- and reconstructs (one copy of) $\ket{\psi}$ from the shares for that color. Meanwhile, an unauthorized set is not in either $\Gamma_0$ nor $\Gamma_1$, so by the security of the underlying 1-copy QSS, the shares for both schemes reveal no information about $\ket{\psi}$. A hybrid argument then shows that the combined shares for an unauthorized set reveal no information about $\ket{\psi}$.

Applying to the case of thresholds, we see that adding additional copies gives us a small improvement over the single copy case. In particular, given $k$ colors/copies, we can give all sets containing user 1 the color 1. Then we give any set containing user 2 (but not 1) the color 2, and so on up until we've used the first $k-1$ colors. Then all sets which contain none of the first $k-1$ users are given the final color $k$. See Figure~\ref{fig:422} for a picture of the $((4,2,2))$-QTSS case.

\begin{figure}[htb]
    \centering
    \begin{minipage}[c]{0.8\textwidth}\centering
        \begin{minipage}[c]{0.25\textwidth}
          \centering
          \includegraphics[width=\textwidth]{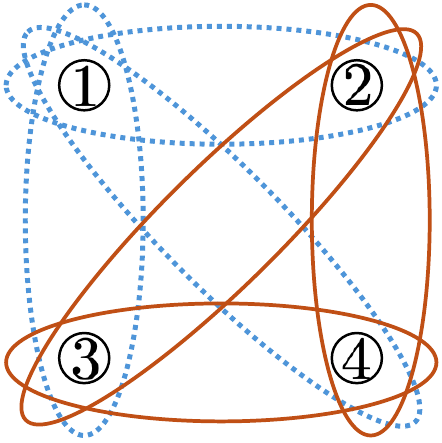}
        \end{minipage}%
        \hspace*{0.25\textwidth}
        \begin{minipage}[c]{0.25\textwidth}
           \centering
           \includegraphics[width=\textwidth]{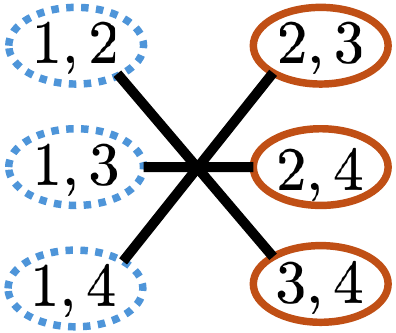}
        \end{minipage}
        \caption{\label{fig:422} Left: the $n=4$ users, and all subsets of $t=2$ users. To share $k=2$ copies, share the first copy to the dotted sets and the second copy to the solid sets. Right: the access structure graph for $t=2$ and $n=4$. Observe how the graph is bipartite.}
    \end{minipage}
\end{figure}

The first $k-1$ colors automatically have no edges between them since any two sets given the same color contain the same user and therefore have no edge. For the last color, this is just a $((n-k+1,t,1))$-QTSS scheme, which is valid as long as $t> (n-k+1)/2$.  

The coloring above is rather simple, and gives a perhaps underwhelming improvement over the 1-copy case. In fact, if we do not care about security, a trivial $((n,t,k))$-QTSS scheme would just give $k-1$ copies of $\ket{\psi}$ to the first $k-1$ users respectively, and then use a $((n-k+1,t,1))$-QTSS scheme for the remaining users.\footnote{The no-cloning obstruction actually applies equally well to QSS schemes with no security. Looking ahead, our Clique Path obstruction also applies equally well to QSS with no security. Also note that \emph{single-copy} QSS for an access structure $\Gamma$ actually implies security for sets in the complement of $\Gamma$. This is because, if an authorized set learns $\ket{\psi}$, by no-cloning, the complement cannot learn any information about $\ket{\psi}$. } A natural question is whether a more sophisticated coloring can give lower thresholds. Unfortunately, Kneser's Conjecture, which was proved by~\cite{lovasz1978kneser}, shows that the above coloring is exactly optimal. Thus, this union-of-access-structures approach cannot achieve thresholds below $(n-k+1)/2$.

\paragraph{Negative Results: The Clique Path Obstruction.} A natural question is whether a more sophisticated scheme can improve on the union-of-access-structures approach to give schemes for graphs that are not $k$-colorable. In particular, the approach above was just using 1-copy schemes as a black box. Perhaps a natively many-copy scheme could use the copies in a more sophisticated way to obtain a richer set of access structures. We were not able to find any such schemes, so the next natural question is whether we can give improved lower-bounds, which lead to Theorems~\ref{thm:neg-thresh-inf} and~\ref{thm:neg-general-inf}.

To illustrate our lower-bound technique, we consider a hypothetical $((5,2,2))$-QTSS scheme for $n=5$ users, threshold $t=2$, using $k=2$ copies. Since user sets $\{1,2\}$ and $\{3,4\}$ are disjoint authorized user sets, we know that $k\geq 2$ copies is necessary. However, the access structure graph requires 3 colors, preventing us from using Theorem~\ref{thm:pos-general-inf} to get a positive result. But since there are not 3 disjoint authorized user sets, the generalized no-cloning obstruction does not rule out $k=2$. 

We now explain how to rule out a $((5,2,2))$-QTSS scheme; the approach readily generalizes to give Theorems~\ref{thm:neg-thresh-inf} and~\ref{thm:neg-general-inf}. We start from a hypothetical $((5,2,2))$-QTSS scheme. Our idea is to pick a pair of disjoint authorized sets, say $\{1,2\}$ and $\{3,4\}$, and try to unambiguously associate each of the two input copies to one of the two sets, one set per input copy. Let's say the first copy goes to $\{1,2\}$ and the second copy goes to $\{3,4\}$. This association would correspond to a coloring for this pair of authorized sets.

Unfortunately, right away this idea runs into problems: for general dealers, it isn't well defined to say that the first copy ``goes to'' $\{1,2\}$. To see why, consider the following two unitaries acting on two identical registers: $U_0$ is the identity, and $U_1$ swaps the registers. We might be tempted to say that $U_0$ associates the first register with the first copy and the second register with the second copy, while $U_1$ does the opposite. However, in our case, the inputs are always in the symmetric subspace since the input is always two copies of the same state. Once we restrict to the symmetric subspace, $U_0$ and $U_1$ are actually both just the identity, and are hence identical.

Hence, any attempt to unambiguously assign copies to authorized sets requires understanding how the dealer behaves outside of the symmetric subspace. But the dealer may just reject anything orthogonal to the symmetric subspace, which would not change the correctness or security of the protocol. For such dealers, just looking at the dealer's isometry would therefore not allow for an unambiguous assignment.

Along similar lines, one may try to determine the assignment by feeding the dealer independent states $\ket{\psi_0}\ket{\psi_1}$ as input, and seeing which state is recovered by $\{1,2\}$ and $\{3,4\}$. But the dealer's first step may be to project onto the symmetric subspace. Conditioned on the projection accepting, the final states recovered by $\{1,2\}$ and $\{3,4\}$ are actually the uniform superposition over $\ket{\psi_0}\ket{\psi_1}$ and $\ket{\psi_1}\ket{\psi_0}$, giving no well-defined way to assign states to each set.

\medskip

Our remedy is the following: we will take the original dealer $D$, and construct an alternate dealer $D'_{\{1,2\},\{3,4\}}$ which does unambiguously define $\{1,2\}$ as getting the first copy, and $\{3,4\}$ as getting the second copy. To construct the new dealer, do the following: run $D$ on an arbitrary input $\ket{\psi}$, and then reconstruct with $\{1,2\},\{3,4\}$. The resulting joint state looks like:

\[\ket{\psi}, \tau_{1,2}\;\;\;,\;\;\;\ket{\psi},\tau_{3,4}\;\;\;,\;\;\;\rho_5\]
Here, $\tau_{1,2}$ (respectively $\tau_{3,4}$) is the ``leftover'' residual state for users $\{1,2\}$ (resp. $\{3,4\}$) after reconstructing $\ket{\psi}$, and $\rho_5$ is user 5's share, which is still around since user 5 did not participate in any reconstruction. Note that $\tau_{1,2},\tau_{3,4},\rho_5$ will typically be entangled. However, they cannot be entangled with the two copies of the pure state $\ket{\psi}$.

Now we observe that the joint system $\tau_{1,2},\tau_{3,4},\rho_5$ must actually be independent of $\ket{\psi}$, in that the dealer always produces the same joint state regardless of the input state. This can be seen by supposing $\ket{\psi_0}$ and $\ket{\psi_1}$ yield different systems, and then computing the final state for $\ket{\psi_0}+\ket{\psi_1}$. If the systems for $\ket{\psi_0}$ and $\ket{\psi_1}$ were different, then the system for $\ket{\psi_0}+\ket{\psi_1}$ will actually be entangled with the reconstructed states, violating the correctness of the secret sharing scheme.

Our new dealer $D'_{\{1,2\},\{3,4\}}$ is defined as follows. On input state $\ket{\psi_0}\ket{\psi_1}$, it constructs the state
\[\ket{\psi_0}, \tau_{1,2}\;\;\;,\;\;\;\ket{\psi_1},\tau_{3,4}\;\;\;,\;\;\;\rho_5\]
Then it reverses the reconstruction procedure for users $\{1,2\}$ and $\{3,4\}$ to get the shares for the users. If $\ket{\psi_0}=\ket{\psi_1}$, $D'_{\{1,2\},\{3,4\}}$ outputs the same state as $D$. However now it is clear that $\{1,2\}$ will always reconstruct $\ket{\psi_0}$ and $\{3,4\}$ will always reconstruct $\ket{\psi_1}$.

\medskip

We can analogously define alternate dealers $D'_{S_0,S_1}$ for any pair of disjoint authorized user sets $S_0,S_1$. We can moreover relate the different dealers. We show that for any $\ket{\psi_0}\ket{\psi_1}$, $D'_{\{1,2\},\{3,4\}}$ and $D'_{\{1,2\},\{3,5\}}$ output exactly the same state. This makes intuitive sense: shares for users 1 and 2 are constructed in the same way in both cases, and shares for users 3,4,5 are equivalent to shares of $\ket{\psi_1}^{\otimes 2}$ according to the original dealer $D$. Actually proving equality, however, requires care, as there could in principle be different entanglement across the $\{1,2\}$ sub-system and the $\{3,4,5\}$ sub-system. Our solution is to consider the alternate dealers applied to \emph{all} symmetric product states, for example $(\alpha\ket{\psi_0}+\beta\ket{\psi_1})^{\otimes 2}$. On such product states, our alternate dealers behave identically to the original dealer, and are therefore all equivalent. Meanwhile, the symmetric product states contain cross terms that correspond exactly to the asymmetric states we are interested in, but also cross terms corresponding to different asymmetric states. By leveraging the unitarity of the alternate dealers, we actually show that the different dealers behave identically for \emph{every} cross term, giving the desired result.

More generally, $D'_{S_0,S_1}$ and $D'_{S_0,S_1'}$ output the same state whenever $S_0$ is disjoint from both $S_1$ and $S_1'$, and symmetrically for substitutions in the first coordinate. Chaining such substitutions, we consider the following sequence of dealers:
\[D'_{\{1,2\},\{3,4\}}\rightarrow D'_{\{1,2\},\{3,5\}}\rightarrow D'_{\{1,4\},\{3,5\}}\rightarrow D'_{\{1,4\},\{2,5\}}\rightarrow D'_{\{3,4\},\{2,5\}}\]

Each of these dealers, on input $\ket{\psi_0}\ket{\psi_1}$, produces the same state. However, observe that from $D'_{\{1,2\},\{3,4\}}$, we have set $\{3,4\}$ recovers $\ket{\psi_1}$, while in $D'_{\{3,4\},\{2,5\}}$, set $\{3,4\}$ recovers $\ket{\psi_0}$. This is impossible if the dealers construct the same state. This gives a contradiction, showing that the original $((5,2,2))$-QTSS scheme is impossible. See Figure~\ref{fig:522}.

\begin{figure}[htb]
    \centering
    \begin{minipage}[c]{0.8\textwidth}
        \centering
        \includegraphics[width=0.3\textwidth]{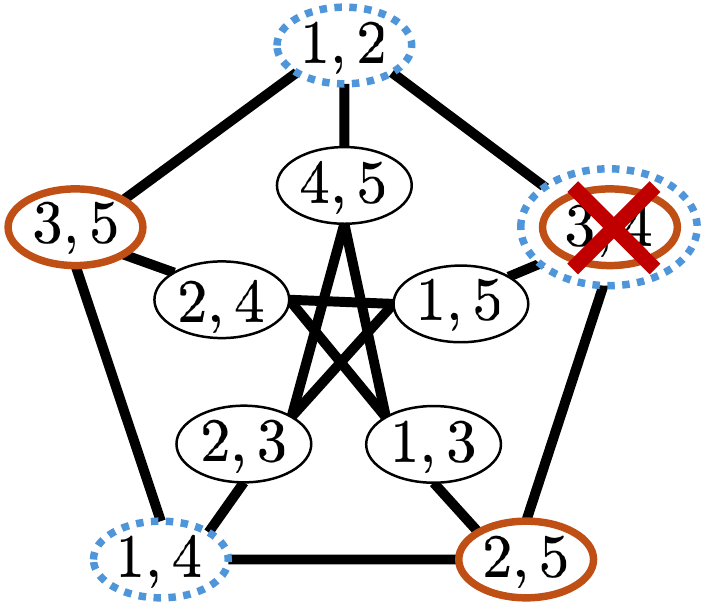}
        \caption{\label{fig:522} The $((5,2,2))$-QTSS case. Start with an arbitrary coloring of the nodes $\{1,2\}$ and $\{3,4\}$, where the coloring corresponds to which position those sets appear in the alternate dealer. Then following the edges around the perimeter, we can color in the nodes $\{3,5\}$, $\{1,4\}$, and $\{2,5\}$, and finally $\{3,4\}$ again. Each edge in this cycle corresponds to an alternate dealer associated to the (ordered) pair of endpoints, and all the dealers in the cycle must produce the same state. But the final coloring for $\{3,4\}$ is different from the first coloring, meaning that $\{3,4\}$ recovers different states under each dealer. This is impossible since all the dealers produce the same state.}
    \end{minipage}
\end{figure}

\medskip

The above approach readily generalizes to give meaningful obstructions for more general access structures. Notice that the sequence of dealers above corresponds to a cycle in the access structure graph:
\[\{3,4\}\rightarrow\{1,2\}\rightarrow\{3,5\}\rightarrow\{1,4\}\rightarrow\{2,5\}\rightarrow\{3,4\}\]
 Each edge corresponds to a dealer. The reason that $\{3,4\}$ starts on the right and ends on the left is exactly because this cycle has odd length. 

 Our more general result is therefore that if the access structure graph $G_\Gamma$ contains an odd-length cycle, then there is no $((n,\Gamma,2))$-QSS scheme. Notice that this exactly characterizes the 2-copy case, as a graph is 2-colorable if and only if it has no odd-length cycle. This proves Theorem~\ref{thm:neg-general-inf} for the case $k=2$, and by specializing to thresholds, we obtain Theorem~\ref{thm:neg-thresh-inf} for $k=2$.

 \medskip

Moving to $k=3$ or more copies, our approach also generalizes, but is not tight. We can readily define alternate dealers $D'_{S_0,S_1,\cdots,S_{k-1}}$ for any $k$ pairwise disjoint authorized user sets; equivalently for any $k$-clique in $G_\Gamma$. We can give $\ket{\psi_0}\ket{\psi_1}\cdots\ket{\psi_{k-1}}$ to this dealer, and then set $S_i$ will recover state $\ket{\psi_i}$. Moreover, we can replace any $S_i$ with $S_i'$, as long as $S_i'$ remains disjoint from the other $S_j$. Generalizing our result above, the resulting dealers produce identical states. In particular, $S_i$ and $S_i'$ must both output $\ket{\psi_i}$.

We can then change multiple sets, one at a time, to obtain paths of $k$-cliques. If we ever obtain a Clique Path that results in one of the $S_i$ appearing but at a different index in the dealer, this corresponds to an inconsistent coloring of $S_i$, and means that $S_i$ recovers a different state $\ket{\psi_j}$ for $j\neq i$ under this dealer. On the other hand, since all the dealers in the Clique Path must produce the same state, we know that $S_i$ must always produce $\ket{\psi_j}$. This is a contradiction, showing that a $((n,\Gamma,k))$-QSS scheme actually cannot exist. See Figure~\ref{fig:723} for an example. This generally proves Theorem~\ref{thm:neg-general-inf}. We then specialize to the threshold case, in which case we observe there is a contradictory Clique Path as long as there is a $k$-clique plus 1 more user; that is as long as $n\geq kt+1$. This proves Theorem~\ref{thm:neg-thresh-inf}.

\begin{figure}[htb]
    \centering
    \begin{minipage}[c]{0.8\textwidth}
        \centering
        \includegraphics[width=0.4\textwidth]{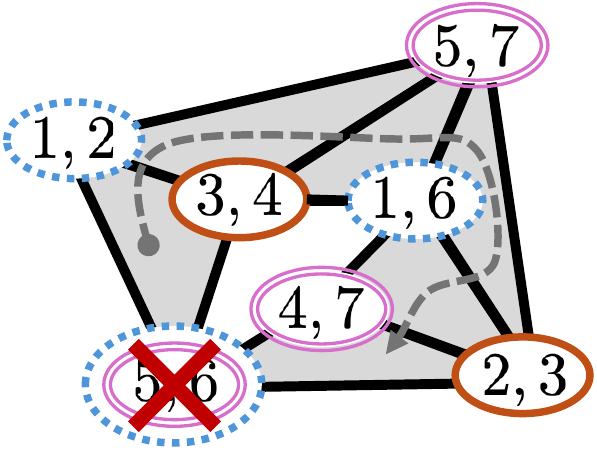}
        \caption{\label{fig:723} The $((7,2,3))$-QTSS case. The triangles are the $3$-cliques, and the dashed arrow is the Clique Path. Observe how the node $\{5,6\}$ is assigned different colors by the first and last Clique. Note that for $((7,2,3))$-QTSS, there are actually shorter Clique Paths (namely, of length 5) which can be used to give a contradiction. However, the above Clique Path more generally rules out $((7,\Gamma,3))$-QSS for \emph{any} access structure where $\{1,2\},\{1,6\},\{2,3\},\{3,4\},\{4,7\},\{5,6\},\{5,7\}$ are authorized sets.}
    \end{minipage}
\end{figure}

\subsection{Further Discussion}\label{sec:discuss}

\paragraph{Entangled inputs.} In this work, we exclusively consider the case of pure-state inputs. In contrast, existing work on (single-copy) quantum secret sharing and error correction often considers the input being a part of a larger system, and the authorized sets recover the sub-system together with any entanglement with the larger system. However, multiple input copies seems incompatible with entanglement-preserving QSS/error correction. Indeed, due to the monogamy of entanglement, given a state $\rho$ entangled with a larger system, it is infeasible to produce two identical copies of $\rho$ that each possess the same entanglement.

Nevertheless, we think the restriction of QSS/error correction to pure states remains interesting and useful, since in many cases we only care about pure states.

\paragraph{Directions for future work.} The $k=1$ case for QSS was resolved by~\cite{cleve1999share,gottesman_solo}, and our work resolves the $k=2$ case, while giving non-tight positive and negative results for $k=3$ and above. A natural conjecture given the close connection between our colorings and our results (both positive and negative) is that coloring captures feasibility for general $k$:

\begin{conjecture}\label{conj}There exists a $((n,\Gamma,k))$-QSS scheme if \emph{and only if} the associated access structure graph is $k$-colorable.
\end{conjecture}

Our work proves Conjecture~\ref{conj} for the case $k\leq 2$. We leave proving or disproving Conjecture~\ref{conj} for $k\geq 3$ as an interesting open question.

Even for $k=3$, we believe that significantly new ideas are required to prove the conjecture. In both~\cite{cleve1999share,gottesman_solo} and our work, the obstruction is proved by constructing an explicit counterexample. In~\cite{cleve1999share,gottesman_solo}, they construct an impossible cloning algorithm from any disallowed access structure. In our work, we produce a sequence of identical states that behave differently. In both cases the counterexample is built from a combinatorial witness that the
access structure graph is not $k$-colorable, and the proof converts that witness
into an information-theoretic contradiction. For $k\leq 2$, such witnesses exist for \emph{any} non-$k$-colorable graph, namely an edge for $k=1$ and an odd-length cycle for $k=2$. However for $k\geq 3$, $k$-coloring is ${\sf NP\text{-}Complete}$, and so unless ${\sf NP}={\sf coNP}$, there is no polynomial-time verifiable witness for being non-$k$-colorable. Even if allowing for quantum witnesses, there will be no polynomial-time-verifiable witness unless ${\sf NP}\subseteq{\sf coQMA}$, which seems unlikely.

Therefore, any proof of Conjecture~\ref{conj}, even for $k=3$, will require either constructing counterexamples requiring super-polynomial time to verify, or else must leverage a different proof that does not work by constructing explicit counterexamples. Perhaps the most natural approach to prove Conjecture~\ref{conj} is to derive from any potential dealer a coloring of the authorized sets. Indeed, this was our initial direction. However, as mentioned above, extracting an unambiguous coloring from the dealer seems difficult, if not impossible. 

\medskip

Another direction for future work is to consider quantum error correction for Pauli errors. The ability to correct $d$ arbitrary Pauli errors allows for correcting $2d$ erasures. As such, the no-cloning theorem imposes a limit of correcting no more than a Pauli error rate of $1/4$. This is in contrast to the classical setting where codes are capable of correcting arbitrarily close to a bit-flip error rate of $1/2$. Can using multiple copies allow for higher error rates?

\medskip

Finally, consider the $((4,2,2))$-QTSS scheme guaranteed by Theorem~\ref{thm:pos-thresh-inf} and depicted in Figure~\ref{fig:422}. Any two users can recover one of the copies of $\ket{\psi}$. But how many users are required to recover both copies? It is not difficult to see that any set of three users which contains user 1 has an authorized pair for both states and can therefore recover $\ket{\psi}^{\otimes 2}$. However, the set $\{2,3,4\}$ is only authorized to recover one copy.

We may instead hope that a $((4,2,2))$-QTSS scheme allows \emph{any} set of three users to recover both copies. Note that a threshold of $3$ is required, by the usual no-cloning obstruction.

More generally, we may wish to construct a $((n,t,k))$-QTSS scheme where higher thresholds allow for recovering more copies, up to the limit of what is not impossible by no-cloning or other obstructions. As illustrated by the $((4,2,2))$-QTSS case, it seems unlikely that approaches like the union-of-access-structures construction would accomplish this, and that the QTSS scheme must more fundamentally operate on the joint input system. Alternatively, perhaps such a QTSS scheme can be shown to be impossible.

\subsection{Publication History}

A very early version of this work appeared as the first author's undergraduate thesis~\cite{chien_augmented}, which was supervised by the third author. The thesis contained the positive results using multiple copies, but did not include the negative results. The negative results were developed for the second author's undergraduate thesis~\cite{Mukherjee}, also advised by the third author. This paper subsumes these theses, and strengthens the negative results.

Subsequent to the first author's thesis,~\cite{TQC:CGLR23} consider QSS for general access structures, but with a focus on efficiency. Their main result is that, assuming post-quantum secure one-way functions, any access structure $\Gamma$ represented by efficient monotone circuits has a $((n,\Gamma,k))$-QSS scheme with $k\leq n$ that is efficient, and in particular has short shares. One limitation of their scheme is that it is only secure against computationally-bounded adversaries.  Their work does not explore the minimal number of copies required to perform QSS.
\section{Preliminaries}

\subsection{Quantum Secret Sharing}

Very roughly, a (quantum) secret sharing scheme enables a dealer to share (quantum) information across a number of users, such that authorized subsets of users can reconstruct the information, but unauthorized subsets of users cannot. We now describe quantum secret sharing more formally.

One distinction we make from the existing literature on quantum secret sharing is that we will assume that the secret information being shared is a pure state, whereas existing definitions sometimes also allow states that are entangled with an outside register. Our use of pure-state secret sharing is necessary for our definition to make sense in the multi-copy setting.

A quantum secret sharing (QSS) scheme is associated with a number $n$ of users and an access structure $\Gamma\subseteq 2^{[n]}$, which is a collection of subsets of $[n]$. We will always assume that $\Gamma$ is \emph{monotone}, meaning that for any $S\in \Gamma$ and any $T\supseteq S$, $T\in\Gamma$ as well. In this work we consider a multi-copy variant, in which case a QSS scheme is also associated with a number $k$ of copies. Such a QSS scheme consists of two algorithms:
\begin{itemize}
    \item $\Deal(\ket{\psi}^{\otimes k})$: takes as input $k$ identical copies of a quantum state $\ket{\psi}$. It outputs a joint state $\rho_{R_1,R_2,\cdots,R_n}$ over $n$ systems $R_1,\cdots,R_n$, with $R_i$ having Hilbert space $\mathcal H_i$.
    \item $\Recon(S,\rho'_{\{R_i\}_{i\in S}})$ takes as input a subset $S\subseteq[n]$ and a joint state $\rho'_{\{R_i\}_{i\in S}}$ over $|S|$ systems $R_i$ for $i\in S$. It outputs a state $\ket{\psi'}$.
\end{itemize}

We say a set $S\subseteq[n]$ is \emph{authorized} if $S\in\Gamma$. A set is \emph{unauthorized} if it is not authorized ($S\notin\Gamma$).

In typical usage, the entity running $\Deal$ will be called the ``dealer.'' The dealer will take as input ($k$ copies of) the secret state $\ket{\psi}$, and will output $\rho_{R_1,R_2,\cdots,R_n}\gets\Deal(\ket{\psi}^{\otimes k})$. It will then give the partial state $\rho_{R_i}$ on system $R_i$ to user $i$, for $i=1,\cdots,n$. $\rho_{R_i}$ will be referred to as user $i$'s ``share.'' Later, a set of authorized users $S$ will pool their shares $\rho_{\{R_i\}_{i\in S}}$, and run $\Recon(S,\rho_{\{R_i\}_{i\in S}})$, which is supposed to recover $\ket{\psi}$. Meanwhile, if an unauthorized set tries to pool their shares, they should learn nothing about $\ket{\psi}$. We now formalize this expected behavior:

\begin{definition}[QSS Correctness] A QSS scheme is \emph{correct} if, for any pure state $\ket{\psi}$ and any authorized set $S\in\Gamma$, if we let $\rho_{R_1,R_2,\cdots,R_n}\gets\Deal(\ket{\psi}^{\otimes k})$ and $\ket{\psi'}\gets\Recon(S,\rho_{\{R_i\}_{i\in S}})$, then $\ket{\psi'}=\ket{\psi}$.\end{definition}
That is, any authorized set of users actually correctly recovers $\ket{\psi}$.

\begin{definition}[QSS Security] A QSS scheme is \emph{secure} if, for any pair of pure states $\ket{\psi^0},\ket{\psi^1}$ and any unauthorized set $S\notin\Gamma$, if we let $\rho^b_{R_1,R_2,\cdots,R_n}\gets\Deal(\ket{\psi^b}^{\otimes k})$, then $\rho^0_{\{R_i\}_{i\in S}}=\rho^1_{\{R_i\}_{i\in S}}$.\end{definition}
That is, any unauthorized set of users cannot learn anything about the secret.

\begin{definition} A $((n,\Gamma,k))$-QSS scheme is a QSS scheme for $n$ users, the access structure $\Gamma$, and $k$ copies of the input, that is both correct and secure.\end{definition}

\paragraph{QSS for thresholds.} We now study the special case of threshold QSS (QTSS). The threshold access structure $\Gamma_{t,n}$ is the collection of all subsets of $[n]$ of size at least $t$.

\begin{definition} A $((n,t,k))$-QTSS scheme is a QSS scheme for $n$ users, the threshold access structure $\Gamma_{t,n}$, and $k$ copies of the input, that is both correct and secure.\end{definition}

\paragraph{A complete characterization for $k=1$.} We now briefly recall known results for $k=1$, which give a complete picture of what is possible and impossible.

\begin{theorem} \label{thm:qss-1} \cite{cleve1999share} 
A $((n,\Gamma,1))$-QSS scheme exists if and only if there are no two disjoint sets in $\Gamma$.
\end{theorem}

For completeness, we briefly sketch the lower-bound showing that a $((n,\Gamma,1))$-QSS scheme does not exist if $\Gamma$ contains a pair of disjoint subsets. If $\Gamma$ contains two sets $S_0,S_1$ that are disjoint, then we can derive a cloning algorithm as follows: given a state $\ket{\psi}$, secret share it to obtain $\rho_{R_1,\cdots,R_n}$. Divide the shares into $\rho_0=\rho_{\{R_i\}_{i\in S_0}}$ and $\rho_1=\rho_{\{R_i\}_{i\in S_1}}$, which by assumption are disjoint systems. Now run $\Recon$ on $\rho_0$ and $\rho_1$; by correctness, each should recover $\ket{\psi}$, giving two copies of the input state $\ket{\psi}$. This violates the no-cloning theorem.

\begin{corollary} A $((n,t,1))$-QTSS scheme exists if and only if $n<2t$.
\end{corollary}
Indeed, $\Gamma_{t,n}$ contains two disjoint authorized sets exactly when
$n\ge 2t$, as witnessed by $\{1,\dots,t\}$ and $\{t+1,\dots,2t\}$.

\paragraph{The no-cloning bound for many copies.} The cloning argument above
generalizes immediately to $k$ copies.

\begin{corollary}\label{cor:nocloning-k}
No $((n,\Gamma,k))$-QSS scheme exists if $\Gamma$ contains $k+1$ pairwise disjoint sets. In the threshold case, no $((n,t,k))$-QTSS scheme exists when $n\ge(k+1)t$.
\end{corollary}

Notice that if $S_0,\dots,S_k$ are pairwise disjoint authorized sets, each can run
$\Recon$ on its own shares and recover $\ket\psi$, producing $k+1$ copies from
the $k$ given to $\Deal$ and violating no-cloning. For $\Gamma_{t,n}$, the sets
$\{1,\dots,t\},\{t+1,\dots,2t\},\dots$ give $k+1$ pairwise disjoint authorized
sets exactly when $n\ge(k+1)t$.

\subsection{Graph Theory}
In order to discuss the union-of-access-structures approach in Section~\ref{sec:positive}, we establish some key definitions and results in graph theory.

\begin{definition}[Vertex Coloring]
    A vertex coloring of a graph is an assignment of colors to its vertices such that no vertices connected by an edge share the same color.
\end{definition}

\begin{definition}[Chromatic Number]
    The chromatic number $\chi(G)$ of a graph $G$ is the minimum number of colors required to produce a valid vertex coloring of $G$.
\end{definition}

A graph with chromatic number $\chi(G) \leq 2$ is called a \textit{bipartite graph}. Bipartite graphs have been extensively studied, and are characterized by the following theorem:

\begin{theorem}{\label{thm:bipartite_odd}}
    A graph is bipartite if and only if it has no cycle (i.e., a closed loop formed by the edges) of odd length.
\end{theorem}

With these definitions, we can now focus on a specific type of graphs especially relevant to QTSS schemes.

\begin{definition}[Access Structure Graph]\label{def:ac_complement}
    The \emph{access structure graph} $G_\Gamma$ of an access structure $\Gamma$
    is the graph whose vertices are the minimal authorized sets of $\Gamma$, with
    two vertices adjacent if and only if the corresponding sets are disjoint.
\end{definition}

\begin{definition}[Kneser Graph]{\label{def:kneser}}
    A Kneser Graph, represented by $K(m,\ell)$, is the graph whose vertices correspond to the $\ell$-element subsets of a set of $m$ elements, and where two vertices are adjacent if and only if the two corresponding subsets are disjoint. 
\end{definition}

For $\Gamma_{t,n}$ the minimal authorized sets are exactly the $t$-subsets of $[n]$, so $G_{\Gamma_{t,n}}$ is the Kneser graph $K(n,t)$. We now state the theorem first proved in \cite{lovasz1978kneser} that we use to prove our main positive result, Theorem \ref{thm:pos-bound} in Section~\ref{sec:positive}.

\begin{theorem}{\label{thm:kneser_conjecture}}
    For $m\ge 2\ell$, the chromatic number of a Kneser graph $K(m,\ell)$ is $\chi(K(m,\ell)) = m - 2\ell + 2$. 
\end{theorem}

Note that for $m < 2\ell$, the graph has no edges so $\chi = 1$.

\begin{remark}
    Lovász originally showed that the chromatic number of the Kneser
graph $K(a,b)$ is $b+2$. However, he uses an alternative parametrization: in his
definition the total number of elements is $2a+b$ and the subsets have size $a$,
so the formula in Theorem~\ref{thm:kneser_conjecture} gives
$(2a+b)-2a+2=b+2$. In this paper we use the more intuitive definition in
Definition~\ref{def:kneser}.
\end{remark} 

\section{Positive Results}{\label{sec:positive}}
We present our positive results. How can we share multiple copies of a quantum secret such that the no-cloning boundary is not violated? A natural starting point is to decompose the global access structure into smaller, single-copy access structures. Under this approach, which we call the ``union-of-access-structures" approach, we share $k$ independent copies of the input secret by running $k$ parallel single-copy schemes, ensuring the allocation of the shares does not violate the no-cloning theorem.

\subsection{The Dual-Copy Case ($k = 2$)}
Let $\Gamma$ be the desired access structure. Recall the access structure graph $G_\Gamma$ from Definition~\ref{def:ac_complement}.

\begin{theorem}\label{thm:k2-positive}
    A $((n,\Gamma,2))$-QSS scheme is realizable via the union-of-access-structures
    approach if and only if the access structure graph $G_\Gamma$ is bipartite
    ($2$-colorable).
\end{theorem}
\begin{proof}
    $(\impliedby)$ If $G_\Gamma$ is bipartite, we can partition its vertex set into two independent sets $V_1$ and $V_2$. Because no two vertices within $V_i$ share an edge, we know that no two authorized sets in $V_i$ are disjoint. Thus the sub-access structures $\Gamma_1$ and $\Gamma_2$ formed by $V_1$ and $V_2$ respectively satisfy the no-cloning theorem. Consequently, by Theorem~\ref{thm:qss-1} there exist valid single-copy schemes to realize $\Gamma_1$ and $\Gamma_2$. Distributing one copy of the secret according to $\Gamma_1$ and the second copy according to $\Gamma_2$ gives a valid $((n,\Gamma,2))$-QSS scheme. Correctness follows, since any authorized set lies in $\Gamma_1$ or $\Gamma_2$ and reconstructs from the corresponding sub-sharing. For security, an unauthorized set of $\Gamma$ is unauthorized for both $\Gamma_1$ and $\Gamma_2$, so by the security of each single-copy scheme, its shares from each sub-sharing are independent of the secret. Since the two sub-sharings are run independently, its joint state is independent of the secret as well. 

    $(\implies)$ Conversely, suppose a valid $((n,\Gamma,2))$-QSS scheme is obtained by this approach, i.e.\ $\Gamma = \Gamma_1 \cup \Gamma_2$ where copy $i$ is shared according to a single-copy scheme for $\Gamma_i$. By Theorem~\ref{thm:qss-1}, each $\Gamma_i$ contains no two disjoint authorized sets. In $G_\Gamma$, the vertices corresponding to $\Gamma_1$ and $\Gamma_2$ must therefore belong to two independent sets, which implies that $G_\Gamma$ is bipartite. 
\end{proof}
As a concrete example, consider a $((4,2,2))$-QTSS scheme. The authorized sets are the pairs $\{1,2\}, \{2,3\}, \allowbreak \{3,1\},  \{1,4\}, \{2,4\}, \{3,4\}$, and its access structure graph is bipartite: it consists of three disjoint edges, such as $\{1,2\}$--$\{3,4\}$. This guarantees a valid construction by partitioning the pairs into $\Gamma_1 = \{\{1,2\}, \{2,3\}, \{3,1\}\}$ and $\Gamma_2 = \{\{1,4\}, \{2,4\}, \{3,4\}\}$.

\subsection{Generalizing to Arbitrary $k$}{\label{sec:gen-to-arb-k}}
We can easily generalize this framework to an arbitrary number of copies $k \geq 2$. In this setting, the existence of a $((n,\Gamma,k))$-QSS scheme reduces to the $k$-colorability of its access structure graph. 

\begin{theorem}\label{thm:acc_comp_k_colorable}
    A $((n,\Gamma,k))$-QSS scheme is realizable via the union-of-access-structures
    approach if and only if $G_\Gamma$ is $k$-colorable.
\end{theorem}
\begin{proof}
    $(\impliedby)$ Fix a proper $k$-coloring and let $V_1,\dots,V_k$ be the
    color classes, with $\Gamma_1,\dots,\Gamma_k$ the sub-access structures
    they generate. No two sets within a class are adjacent, hence none are
    disjoint, so by Theorem~\ref{thm:qss-1} each $\Gamma_i$ admits a
    single-copy scheme. Share copy $i$ according to $\Gamma_i$ and give each
    user their shares from all $k$ schemes. Correctness follows since any authorized set lies in some
    $\Gamma_i$ and reconstructs from that sub-sharing. An unauthorized set of
    $\Gamma$ is unauthorized for every $\Gamma_i$, so its shares from each
    sub-sharing are independent of the secret. As the $k$ sub-sharings are run
    independently, its joint state is independent of the secret. Hence security also holds.

    $(\implies)$ If $\Gamma=\Gamma_1\cup\cdots\cup\Gamma_k$ with copy $i$
    shared by a single-copy scheme for $\Gamma_i$, then by
    Theorem~\ref{thm:qss-1} no $\Gamma_i$ contains two disjoint sets, so each
    corresponds to an independent set of $G_\Gamma$. These $k$ independent
    sets cover the vertices, giving a proper $k$-coloring.
\end{proof}

We now specialize to the threshold regime. By framing the problem as a graph coloring task, we can establish the exact boundary on the number of participants $n$ for which a $((n,t,k))$-QTSS scheme can exist. The access structure graph for a $((n,t,k))$-QTSS scheme is exactly the Kneser graph $K(n,t)$ (Definition~\ref{def:kneser}), where the vertices are $t$-element subsets of an $n$-element set, adjacent if and only if their corresponding subsets are disjoint. 

Using Lovász's determination of the chromatic number of a Kneser graph (Theorem~\ref{thm:kneser_conjecture}), we obtain our main result:

\begin{theorem}\label{thm:pos-bound}
    A $((n,t,k))$-QTSS scheme is realizable by the union-of-access-structures approach if and only if 
    \[
        n \leq 2t - 2 + k
    \]
\end{theorem}
\begin{proof}
    By Theorem \ref{thm:acc_comp_k_colorable}, we have that a $((n,t,k))$-QTSS scheme is realizable by the union-of-access-structures approach if and only if the access structure graph $K(n,t)$ is $k$-colorable, which requires $\chi(K(n,t)) \leq k$. Substituting Lovász's result (Theorem~\ref{thm:kneser_conjecture}), we have 
    \[
        n - 2t + 2 \leq k\implies n \leq 2t - 2 + k
    \]
    This completes the proof.
\end{proof}

\textit{Remark:} This combinatorial bound shows that the $k = 1$ bound is $n < 2t$ (or equivalently, $n \leq 2t - 1$), whereas adding $k$ copies of the state allows the network size to scale linearly with $k$ without violating no-cloning constraints.

\section{Negative Results}
In Section~\ref{sec:positive}, we determined exactly which schemes are realizable by the union-of-access-structures approach: those whose access structure graph $G_\Gamma$ is $k$-colorable. However, that bound is only a statement about one construction. The question remains---is there some cleverer protocol that evades this bound, perhaps one that does not hard-partition the $k$ copies among $k$ sub-access structures and instead utilizes a jointly entangled encoding with recovery operations acting coherently across copies?

In this section, we show that a portion of the bound is not just an artifact of the construction. We give a criterion, phrased as a color propagation on
$G_\Gamma$, which certifies that \emph{no scheme} realizes a given
access structure. The criterion is purely combinatorial: we seed a set of $k$ pairwise disjoint authorized sets with the labels $1,\dots,k$,
propagate labels by elimination, and look for a conflict. The underlying
argument uses nothing about the scheme beyond unitarity of
reconstruction.
 
We then apply the criterion to threshold structures and recover, in a few lines,
the impossibility of $((kt+1,t,k))$-QTSS.
 
Throughout this section we assume $k\ge 2$; for $k=1$ there is only one label
and the criterion is vacuous, consistent with Theorem~\ref{thm:qss-1}.

\subsection{Canonical Form}\label{sec:canonical}

It is convenient to replace $\Deal$ and $\Recon$ by isometries. 

\begin{lemma}[Purified dealer]\label{lem:neg-purified}
Any $((n,\Gamma,k))$-QSS scheme may be assumed to have $\Deal$ given by an isometry
\[
    \hat U_D \;:\; \mathcal H^{\otimes k}\;\longrightarrow\;
    \mathcal H_1\otimes\cdots\otimes\mathcal H_n\otimes\mathcal H_E ,
\]
where $\mathcal H\cong\mathbb C^d$ is the secret space and $\mathcal H_E$ is a
register retained by no user. We write
$\ket{\Psi_\psi}:=\hat U_D\ket{\psi}^{\otimes k}$ for the global pure state, so
that $\rho_{R_1,\dots,R_n}=\Tr_E\ket{\Psi_\psi}\bra{\Psi_\psi}$.
\end{lemma}

\begin{proof}
$\Deal$ is a channel, so take any Stinespring dilation and let $\mathcal H_E$ be
the dilating register. Correctness and security refer only to the users'
registers, and are hence unaffected.
\end{proof}
 
\begin{lemma}[Unitary reconstruction]\label{lem:neg-unitary}
Let $S\in\Gamma$ and write $\mathcal H_S=\bigotimes_{i\in S}\mathcal H_i$
(enlarged by a private ancilla in a fixed state, which we absorb into
$\mathcal H_S$). Then there is a unitary $\hat U_S$ on $\mathcal H_S$ and a
factorization $\mathcal H_S\cong\mathcal K_S\otimes\mathcal R_S$ with
$\dim\mathcal K_S=d$, such that for every pure $\ket\psi$,
\[
    \big(\hat U_S\otimes I\big)\ket{\Psi_\psi}
      \;=\; \ket{\psi}_{\mathcal K_S}\otimes\ket{\Theta^S_\psi} ,
\]
where $\ket{\Theta^S_\psi}$ is a pure state on $\mathcal R_S$ together with all
registers outside $S$.
\end{lemma}
 
\begin{proof}
Dilate $\Recon(S,\cdot)$ to an isometry on $\mathcal H_S$ and absorb its ancilla
and output space into $\mathcal H_S$, so that it becomes a unitary $\hat U_S$
with a distinguished output factor $\mathcal K_S$ of dimension $d$. By
correctness the reduced state on $\mathcal K_S$ after applying $\hat U_S$ is
$\ket\psi\bra\psi$, which is pure. A subsystem of a pure global state is pure
only if it is unentangled with its complement, so the global state factorizes as
claimed.
\end{proof}
 
We call $\mathcal K_S$ the \emph{output register} of $S$ and $\mathcal R_S$ its
\emph{residue}. Note $\hat U_S$ acts only on registers indexed by $S$; in
particular, if $S$ and $S'$ are disjoint then $\hat U_S$ and $\hat U_{S'}$ act
on disjoint registers and therefore commute.

\subsection{Labeled Cliques and the Residual State}
 
Recall from Section~\ref{sec:positive} that the vertices of the access structure graph
$G_\Gamma$ are the minimal authorized sets, adjacent when disjoint. A set of $k$
pairwise disjoint authorized sets is therefore a $k$-clique of $G_\Gamma$. For
our purposes the clique must be \emph{ordered}, because the argument turns on
which set reads which input copy.
 
\begin{definition}[Labeled Clique]\label{def:neg-clique}
A \emph{labeled $k$-clique} is an ordered tuple $\vec A=(A_1,\dots,A_k)$ of
pairwise disjoint authorized sets. We call $i$ the \emph{index} of $A_i$ in
$\vec A$, and think of the index as a color: $\vec A$ colors $A_i$ with $i$.
\end{definition}
 
A labeled $k$-clique exists if and only if $G_\Gamma$ contains a $k$-clique,
and everything in this section is vacuous otherwise. Given one, we may apply all
$k$ recovery unitaries at once, since they act on disjoint registers. The
following lemma is the engine for the entire section: once every copy has been
reconstructed, what is left over carries no information about the secret.

\begin{lemma}[Constant residue]\label{lem:neg-residue}
Let $\vec A$ be a labeled $k$-clique and set
$\hat U_{\vec A}:=\hat U_{A_1}\hat U_{A_2}\cdots\hat U_{A_k}$. Then there is a
single fixed unit vector $\ket{\mathrm{Rem}}$, independent of the secret, with
\[
    \hat U_{\vec A}\,\ket{\Psi_\psi}
      \;=\; \ket{\psi}_{\mathcal K_{A_1}}\otimes\cdots\otimes
             \ket{\psi}_{\mathcal K_{A_k}}\otimes\ket{\mathrm{Rem}}
\]
for every pure $\ket\psi$, where $\ket{\mathrm{Rem}}$ lives on
$\mathcal R_{A_1}\otimes\cdots\otimes\mathcal R_{A_k}$ together with the
registers of the users outside $A_1\cup\cdots\cup A_k$ and the register
$\mathcal H_E$.
\end{lemma}
\begin{proof}
    First, notice that the $\hat U_{A_i}$ commute and each acts on registers untouched by the others. Hence, we can apply Lemma~\ref{lem:neg-unitary} to each $A_i$ in turn to obtain the stated factorization, except with $\ket{\mathrm{Rem}}$ replaced by some $\ket{\mathrm{Rem}_\psi}$. It remains to show this vector does not depend on
    $\psi$.

    The composite $V:=\hat U_{\vec A}\,\hat U_D$ is an isometry, so it preserves
inner products. Let $\ket{\psi^0},\ket{\psi^1}$ be non-orthogonal. Comparing the
inner product of the inputs $\ket{\psi^0}^{\otimes k},\ket{\psi^1}^{\otimes k}$
with that of the outputs gives
\[
    \braket{\psi^0|\psi^1}^{\,k}
      \;=\; \braket{\psi^0|\psi^1}^{\,k}\,
            \braket{\mathrm{Rem}_{\psi^0}|\mathrm{Rem}_{\psi^1}} .
\]
Since $\braket{\psi^0|\psi^1}\neq 0$ we may divide, obtaining
$\braket{\mathrm{Rem}_{\psi^0}|\mathrm{Rem}_{\psi^1}}=1$. Two unit vectors with
inner product $1$ are equal. Any two states of $\mathcal H$, orthogonal or not,
are connected by a chain of pairwise non-orthogonal states, so
$\ket{\mathrm{Rem}_\psi}$ is the same vector for all $\psi$.
\end{proof}

Because $\ket{\mathrm{Rem}}$ is a fixed vector, we may run the whole dealing process backwards on arbitrary, not necessarily identical, inputs. This allows us to define an alternate dealer, as follows.

\begin{definition}[Alternate dealer]\label{def:neg-altdealer}
Let $\vec A$ be a labeled $k$-clique. The \emph{alternate dealer}
$\mathcal D'_{\vec A}$ is the linear map
\[
    \mathcal D'_{\vec A}\;:\;\mathcal H^{\otimes k}\longrightarrow
      \mathcal H_1\otimes\cdots\otimes\mathcal H_n\otimes\mathcal H_E,
    \quad
    \mathcal D'_{\vec A}(\ket{\phi_1}\otimes\cdots\otimes\ket{\phi_k})
      := \hat U_{\vec A}^{\dagger}
         \Big(\ket{\phi_1}\otimes\cdots\otimes\ket{\phi_k}\otimes\ket{\mathrm{Rem}}\Big),
\]
where $\ket{\phi_i}$ is placed in the output register $\mathcal K_{A_i}$.
\end{definition}

Two key properties follow immediately from Lemma~\ref{lem:neg-residue} and the
definition. 
 
\begin{lemma}\label{lem:neg-altprops}
For every labeled $k$-clique $\vec A$:
\begin{enumerate}
    \item \textbf{(Agreement on the diagonal.)}
      $\mathcal D'_{\vec A}\big(\ket\psi^{\otimes k}\big)=\ket{\Psi_\psi}$ for
      every pure $\ket\psi$. In particular, all labeled cliques agree on
      identical inputs.
    \item \textbf{(Index extraction.)} For every $i$ and \emph{all} inputs
      $\ket{\phi_1},\dots,\ket{\phi_k}$,
      \[
        \hat U_{A_i}\,\mathcal D'_{\vec A}(\ket{\phi_1}\otimes\cdots\otimes\ket{\phi_k})
          \;=\; \ket{\phi_i}\otimes(\cdots),
      \]
      with the omitted factor lying on $\mathcal R_{A_i}$ and the registers
      outside $A_i$. That is, index $i$ is tied to the set $A_i$, and
      $\hat U_{A_i}$ recovers whatever was fed into it. This holds even when the
      inputs are not identical.
\end{enumerate}
\end{lemma}

Property (1) says every labeled clique reproduces the honest dealer on the
diagonal. Property (2) says the labeled cliques may nonetheless be told apart
off the diagonal, by asking which set reads which input.

\subsection{The Substitution Lemma}
\begin{definition}
Let $\vec A$ be a labeled $k$-clique, let $i_0\in[k]$, and let $A'$ be an
authorized set disjoint from $A_j$ for every $j\neq i_0$. The labeled clique
$\vec B$ obtained from $\vec A$ by replacing $A_{i_0}$ with $A'$, leaving all
other entries fixed, is said to be obtained from $\vec A$ by a
\emph{substitution at index $i_0$}.
\end{definition}

\begin{lemma}[Substitution]\label{lem:neg-substitution}
If $\vec B$ is obtained from $\vec A$ by a substitution, then
$\mathcal D'_{\vec A}=\mathcal D'_{\vec B}$ as linear maps. That is, they agree
on all inputs, not only on the diagonal.
\end{lemma}

\begin{proof}
    Write $D:=\mathcal D'_{\vec A}-\mathcal D'_{\vec B}$ and let $i_0$ be the index at which the cliques differ. By Lemma~\ref{lem:neg-altprops}(1), $D(\ket\psi^{\otimes k})=0$ for every unit vector $\ket\psi$, hence by homogeneity $D(\ket v^{\otimes k})=0$ for every $\ket v\in\mathcal H$.

Fix an orthonormal basis $\ket{e_1},\dots, \ket{e_d}$ of $\mathcal{H}$ and let
$\ket{v_{\vec z}}=\sum_{a=1}^{d}z_a\ket{e_a}$ for $\vec z\in\mathbb C^d$. Expanding
$D\big(\ket{v_{\vec z}}^{\otimes k}\big)=0$ multilinearly gives the polynomial
identity
\[
    \sum_{w\in[d]^k} z_{w_1}z_{w_2}\cdots z_{w_k}\;
      D\big(\ket{e_{w_1}}\otimes\cdots\otimes \ket{e_{w_k}}\big)\;=\;0
    \qquad\text{for all }\vec z\in\mathbb C^d .
\]

Notice that the monomial attached to a word $w \in [d]^k$ depends only on its
\emph{type} $\mathrm{type}(w)\in\mathbb N^d$, the vector recording how many
times each letter occurs in $w$. Equating the coefficient of each monomial to zero therefore gives 
\begin{equation}\label{eq:neg-type-sum}
    \sum_{w\,:\,\mathrm{type}(w)=\vec m}
      D\big(\ket{e_{w_1}}\otimes\cdots\otimes \ket{e_{w_k}}\big)\;=\;0 
\end{equation}
for every type $\vec m$.

Write $C_j:=A_j=B_j$ for $j\neq i_0$ for the $k-1$ sets shared by the two
cliques, and set $W:=\bigotimes_{j\neq i_0}\hat U_{C_j}$. The $\hat U_{C_j}$
act on disjoint registers, so they commute, and applying $W$ cancels the
corresponding adjoints in Definition~\ref{def:neg-altdealer} for \emph{both}
$\mathcal D'_{\vec A}$ and $\mathcal D'_{\vec B}$, the remaining
factor acting only on registers indexed by $A_{i_0}$ or $B_{i_0}$, which are
disjoint from every $C_j$. Hence by
Lemma~\ref{lem:neg-altprops}(2), for any word $w \in [d]^k$,
\[
    W\,D\big(\ket{e_{w_1}}\otimes\cdots\otimes \ket{e_{w_k}}\big)
    \;=\;\Big(\bigotimes_{j\neq i_0}\ket{e_{w_j}}\Big)\otimes\ket{\delta_w},
\]
where $\ket{\delta_w}$ is a vector on the remaining registers.

Apply $W$ to \eqref{eq:neg-type-sum}. Within a fixed type $\vec m$,
the restricted word $w|_{[k]\setminus\{i_0\}}$ determines the missing letter
$w_{i_0}$, hence determines $w$. This means that distinct terms of the sum carry distinct restricted words, and the vectors $\bigotimes_{j\neq i_0}\ket{e_{w_j}}$ appearing in them are linearly independent, since the $\ket{e_a}$ are orthonormal. Every term must therefore vanish separately: $\ket{\delta_w}=0$, and hence $D(\ket{e_{w_1}}\otimes\cdots\otimes \ket{e_{w_k}})=0$ for every word $w\in[d]^k$. These product vectors span $\mathcal H^{\otimes k}$, so $D=0$, as required. 
\end{proof}

\subsection{The Clique Path Obstruction}
\begin{definition}[Clique Path]\label{def:neg-cliquepath}
A \emph{Clique Path} is a sequence
$P=\big(\vec A^{(0)},\vec A^{(1)},\dots,\vec A^{(\ell)}\big)$ of labeled
$k$-cliques in which each $\vec A^{(m+1)}$ is obtained from $\vec A^{(m)}$ by a
substitution.
\end{definition}

\begin{corollary}\label{cor:neg-path-equal}
All labeled cliques along a Clique Path have the same alternate dealer:
$\mathcal D'_{\vec A^{(0)}}=\cdots=\mathcal D'_{\vec A^{(\ell)}}$.
\end{corollary}
 
\begin{proof}
Induction on $\ell$ using Lemma~\ref{lem:neg-substitution}.
\end{proof}

A Clique Path colors authorized sets: the set at index $i$ of some
clique along the path receives color $i$. The obstruction is that this coloring
may fail to be well defined.
 
\begin{definition}[Inconsistent Clique Path]\label{def:neg-inconsistent}
A Clique Path $P$ \emph{colors} an authorized set $S$ with $i$ if
$A^{(m)}_i=S$ for some $m$. We say $P$ is \emph{inconsistent} if it colors some
authorized set with two distinct indices, i.e.\ if there are $m,m'$ and
$i\neq i'$ with $A^{(m)}_i=A^{(m')}_{i'}$.
\end{definition}

\begin{theorem}[Clique-Path obstruction]\label{thm:neg-criterion}
Let $\Gamma$ be a monotone access structure and $k\ge2$. If there is an
inconsistent Clique Path for $\Gamma$, then no $((n,\Gamma,k))$-QSS scheme exists.
\end{theorem}

\begin{proof}
Suppose a scheme exists, and put it in the canonical form of Section~\ref{sec:canonical}. Let
$P$ be an inconsistent Clique Path, witnessed by an authorized set $S$ with
$A^{(m)}_i=S=A^{(m')}_{i'}$ for $i\neq i'$. Write $\vec B:=\vec A^{(m)}$ and
$\vec B':=\vec A^{(m')}$. By Corollary~\ref{cor:neg-path-equal},
$\mathcal D'_{\vec B}=\mathcal D'_{\vec B'}$ as linear maps.
 
Choose inputs $\ket{\phi_1},\dots,\ket{\phi_k}$ with $\ket{\phi_i}$ and $\ket{\phi_{i'}}$ linearly
independent, and let
\[
    \ket{\Psi^\ast}\;:=\;\mathcal D'_{\vec B}(\ket{\phi_1}\otimes\cdots\otimes\ket{\phi_k})
      \;=\;\mathcal D'_{\vec B'}(\ket{\phi_1}\otimes\cdots\otimes\ket{\phi_k}),
\]
a single state. Now let $S$ reconstruct, by applying $\hat U_S$. Reading through
$\vec B$, the set $S$ sits at index $i$, so
Lemma~\ref{lem:neg-altprops}(2) gives
$\hat U_S\ket{\Psi^\ast}=\ket{\phi_i}\otimes\ket{X}$. Reading through $\vec B'$, the
set $S$ sits at index $i'$, so the same lemma gives
$\hat U_S\ket{\Psi^\ast}=\ket{\phi_{i'}}\otimes\ket{Y}$. Hence
\[
    \ket{\phi_i}\otimes\ket{X}\;=\;\ket{\phi_{i'}}\otimes\ket{Y}
\]
in $\mathcal K_S\otimes\mathcal R_S\otimes\mathcal H_{\overline S}\otimes\mathcal H_E$, where $\mathcal H_{\overline S}$ denotes the registers outside $S$. The state $\hat{U}_S\ket{\Psi^\ast}$ is thus a product vector whose first factor is proportional to $\ket{\phi_i}$ and also to $\ket{\phi_{i'}}$. Since these are linearly independent, it must be the zero vector. This is a contradiction, since $\hat{U}_S$ is a unitary and $\ket{\Psi^\ast}$ is a unit vector, hence $ \|\hat{U}_S\ket{\Psi^\ast}\| = 1 \neq 0$. 
\end{proof}

\begin{remark}[Only the sets named by the path matter]\label{rem:neg-local}
The proof uses no authorized set other than those appearing along $P$. Hence a
single inconsistent Clique Path rules out \emph{every} access structure
containing the sets it names: if $P$ uses the authorized sets
$S^{(1)},\dots,S^{(r)}$, then no $((n,\Gamma,k))$-QSS scheme exists for any monotone
$\Gamma$ with $S^{(1)},\dots,S^{(r)}\in\Gamma$. In particular, the obstruction is
not specific to threshold structures.
\end{remark}

\paragraph{Searching by color propagation.} In practice, we look for an
inconsistent path by propagating colors and eliminating. Fix a starting
labeled clique $\vec A^{(0)}$ and color its sets $1,\dots,k$. Then repeat:
 
\begin{quote}
\textbf{Propagation step.} Let $\vec B$ be a labeled clique already reached,
let $i\in[k]$, and let $v$ be any authorized set disjoint from $B_j$ for all
$j\neq i$. Substituting $v$ at index $i$ reaches a new labeled clique. We record
color $i$ for $v$.
\end{quote}
 
Read as a coloring rule: a vertex of $G_\Gamma$ adjacent to the $k-1$ sets of
an already-reached clique other than the one at index $i$ is forced to take the
remaining color $i$. If some vertex ever receives two colors, the path
traced to it is inconsistent and Theorem~\ref{thm:neg-criterion} applies. Note
the $k-1$ neighbors must be $k-1$ sets of a \emph{single} already-reached
clique, not merely $k-1$ colored sets found anywhere in the graph.

\begin{remark}[Recovering the no-cloning bound]
For $k \geq 2$, the obstruction subsumes the no-cloning obstruction of
Theorem~\ref{thm:qss-1}. Suppose $G_\Gamma$ contains a $(k+1)$-clique
$A_1,\dots,A_{k+1}$. Starting from $(A_1,\dots,A_k)$, substitute $A_{k+1}$ at
index $1$ and then $A_1$ at index $2$; both substitutions are legal, since all
$k+1$ sets are pairwise disjoint. The resulting Clique Path colors $A_1$ with
both $1$ and $2$, so no scheme exists. Thus $k+1$ pairwise disjoint authorized
sets give an inconsistent Clique Path.
\end{remark}

\begin{remark}[The obstruction requires $\chi(G_\Gamma)>k$]\label{rem:neg-soundness}
The criterion never contradicts the positive results. Suppose $G_\Gamma$ is
$k$-colorable and fix a proper coloring $\chi$. If $\chi(G_\Gamma)<k$ then
$G_\Gamma$ has no $k$-clique and there is nothing to prove, so assume exactly
$k$ colors are used. The $k$ sets of a labeled clique are pairwise disjoint,
hence pairwise adjacent, and so carry $k$ distinct $\chi$-colors, leaving none spare. A substitution at index $i$ installs an
authorized set disjoint from the other $k-1$, which must therefore avoid their
$k-1$ $\chi$-colors, and the only color remaining is that of the set it just
replaced. Thus the tuple $(\chi(A_1),\dots,\chi(A_k))$ is invariant along any
Clique Path, and no set receives two indices. So whenever a scheme exists by
the union-of-access-structures approach, no inconsistent Clique Path exists.
\end{remark}

\paragraph{The case $k=2$.} Here the criterion is exactly the odd-cycle
condition, and hence tight.

\begin{proposition}\label{prop:neg-k2}
Let $k=2$. There is an inconsistent Clique Path for $\Gamma$ if and only if
$G_\Gamma$ contains a cycle of odd length.
\end{proposition}
 
\begin{proof}
A labeled $2$-clique is an ordered edge of $G_\Gamma$, and a substitution replaces one endpoint by a new neighbor of the other. Listing vertices in the order they enter gives a walk $w_0, w_1, \dots $ whose consecutive pairs are the cliques of the path. Substitutions alternate between the two indices, so $w_j$'s index depends only on the parity of $j$. A set therefore receives two colors exactly when it occupies positions of opposite parity, i.e. when the walk revisits it after an odd number of steps. A graph has an odd closed walk iff it has an odd cycle.
\end{proof}
 
Combining Proposition~\ref{prop:neg-k2} with Theorem~\ref{thm:neg-criterion} and
the positive result of Section~\ref{sec:positive}, we obtain a complete characterization for two
copies, for arbitrary monotone access structures.
 
\begin{theorem}\label{thm:neg-k2-tight}
For any monotone $\Gamma$, an $((n,\Gamma,2))$-QSS scheme exists if and only if
$G_\Gamma$ is bipartite.
\end{theorem}

\subsection{The Clique Path Obstruction Applied to Threshold Structures}
 
We now apply the obstruction to QTSS---the case $\Gamma_{t,n}$ with $n=kt+1$. Here the
substitutions take a particularly simple form.

\begin{lemma}\label{lem:neg-sliding}
Let $n=kt+1$. Any labeled $k$-clique of $G_{\Gamma_{t,n}}$ can be reached from any
other by a Clique Path.
\end{lemma}

\begin{proof}
    A labeled $k$-clique consists of $k$ pairwise disjoint $t$-subsets of $[kt+1]$, occupying $kt$ users and leaving exactly one user free. Call that user the \emph{hole}. A substitution at index $i$ must install a $t$-set disjoint from the other $k-1$ sets, and the only users available to it are those of $A_i$ itself together with the hole. So every substitution simply exchanges the hole with one member of one set, and any user can be moved into any set this way, one exchange at a time.
     
    Now fix a target labeled $k$-clique. Call a user \emph{misplaced} if its
    index in the current clique differs from its index in the target, counting
    the hole as an index of its own, and let $h$ denote the user currently in
    the hole. If $h$ is misplaced, with target index $i$, exchange the hole
    with a misplaced member of $A_i$; such a member exists, since $A_i$ is full
    and does not contain $h$. This places $h$ correctly and moves into the hole
    a user that was already misplaced, so the number of misplaced users
    strictly decreases. If instead $h$ is correctly placed but some user is
    misplaced, one exchange moves that user into the hole and returns us to the
    previous case. Iterating, the number of misplaced users eventually reaches zero, at
    which point the current clique equals the target.
\end{proof}

Figure~\ref{fig:723} shows an inconsistent Clique Path for $t=2$, $k=3$, $n=7=kt+1$.

\begin{theorem}\label{thm:neg-allbig}
No $((n,t,k))$-QTSS scheme exists for any $n\ge kt+1$ and $k\ge2$.
\end{theorem}
\begin{proof}
    First we will show that the claim holds for $n = kt+1$. Since $kt\le kt+1$, labeled $k$-cliques exist in $G_{\Gamma_{t,kt+1}}$. Take any
    $\vec A=(A_1,A_2,\dots,A_k)$. By Lemma~\ref{lem:neg-sliding} there is a Clique
    Path from $\vec A$ to $(A_2,A_1,A_3,\dots,A_k)$, which colors $A_1$ with both
    $1$ and $2$. The path is inconsistent, so Theorem~\ref{thm:neg-criterion}
    applies and we conclude that no $((kt+1,t,k))$-QTSS scheme exists.
    
    Now suppose $n > kt+1$, and assume for contradiction that a $((n,t,k))$-QTSS scheme
exists. Define a scheme for the $kt+1$ users $1,\dots,kt+1$ by running the
original dealer and absorbing the shares of users $kt+2,\dots,n$ into
$\mathcal H_E$. Correctness holds, since any $t$ of the remaining users
form an authorized set of $\Gamma_{t,n}$. Security also holds, since any
$t-1$ of them form an unauthorized set of $\Gamma_{t,n}$, so their joint state
is independent of $\ket\psi$. This is a $((kt+1,t,k))$-QTSS scheme, contradicting the
first part of the proof.
\end{proof}

\subsection{What the Obstruction Does Not Decide}
 
For threshold structures, combining Theorem~\ref{thm:pos-bound} with Theorem~\ref{thm:neg-allbig}:
\[
\begin{array}{ll}
  n\le 2t-2+k          & \text{realizable (union of access structures);}\\
  2t-1+k\le n\le kt    & \text{undecided;}\\
  n\ge kt+1            & \text{unrealizable.}
\end{array}
\]
The middle range is nonempty exactly when $(k-2)(t-1)\ge1$; in particular the
bounds meet for every $t$ when $k=2$, consistent with
Theorem~\ref{thm:neg-k2-tight}, and the smallest undecided case is
$((6,2,3))$-QTSS.
 
Why does $((6,2,3))$ resist our method? First, we see that the
Clique Path obstruction does not apply here. With $n=6$, $t=2$, $k=3$, a
labeled $3$-clique is a perfect matching of $[6]$ into three pairs. There is no hole, and a substitution at index $i$ would have to install a $2$-subset of the two users freed by deleting $A_i$, which is $A_i$ again. No substitution is available at all, every Clique Path has length zero, and no set ever receives a second color. More generally, an inconsistent path requires both that $k$ disjoint authorized sets exist and that there be room to vary one of them, which for threshold structures means that $n\ge kt+1$. The boundary in Theorem~\ref{thm:neg-allbig} is exactly the point at which substitution becomes possible.
 
Second, $G_{\Gamma_{2,6}}=K(6,2)$ has $\chi=6-4+2=4>3$, so no scheme of the
union-of-access-structures type exists. However, a failure of $k$-colorability does not produce an inconsistent Clique Path. The obstruction requires substitutions to be available, and there are none here. Closing the gap therefore requires either a construction that is not based on the union of access structures, or an impossibility argument that does not proceed by substituting one authorized set for another.

\bibliographystyle{alpha}
\bibliography{bib}

\end{document}